\documentclass[sigconf]{acmart}

\AtBeginDocument{%
  \providecommand\BibTeX{{%
    \normalfont B\kern-0.5em{\scshape i\kern-0.25em b}\kern-0.8em\TeX}}}

\setcopyright{rightsretained}
\renewcommand\footnotetextcopyrightpermission[1]{} %
\usepackage{soul}
\usepackage{color, colortbl}
\usepackage{multirow}
\usepackage{algorithm}  
\usepackage{algorithmicx}  
\usepackage[noend]{algpseudocode} 
\usepackage{amsmath}  
\usepackage{bm}
\usepackage{subcaption}

\usepackage{amssymb}

\newtheorem{theorem}{Theorem}
\newtheorem{lemma}[theorem]{Lemma}

\makeatletter

\renewcommand\subsubsection{\@secnumfont}{\bfseries}%
\renewcommand\subsubsection{\@startsection{subsubsection}{3}
  \z@{.5\linespacing\@plus.7\linespacing}{-.5em}%
  {\normalfont\bfseries}}
  
  \makeatother

\begin{document}

\title{Celeritas: Fast Optimizer for Large Dataflow Graphs}

\author{Hengwei Xu}
\email{xuhw@mail.ustc.edu.cn}
\affiliation{University of Science and Technology of China\city{Hefei}
\country{China}}
\author{Yong Liao}
\email{yliao@ustc.edu.cn}
\affiliation{University of Science and Technology of China\city{Hefei}
\country{China}}
\author{Haiyong Xie}
\email{haiyong.xie@ieee.org}
\affiliation{Adv. Innovation Center for Human Brain Protection, Capital Medical University\city{Beijing}
\country{China}}
\author{Pengyuan Zhou}\authornote{*Corresponding authors}
\email{pzhou@ustc.edu.cn}
\affiliation{University of Science and Technology of China\city{Hefei}
\country{China}}

\begin{abstract}

The rapidly enlarging neural network models are becoming increasingly challenging to run on a single device. Hence model parallelism over multiple devices is critical to guarantee the efficiency of training large models. Recent proposals fall short either in long processing time or poor performance. 
Therefore, we propose Celeritas, a fast framework for optimizing device placement for large models. Celeritas employs a simple but efficient model parallelization strategy in the Standard Evaluation, and generates placement policies through a series of scheduling algorithms. 
We conduct experiments to deploy and evaluate Celeritas on numerous large models. The results show that Celeritas not only reduces the placement policy generation time by 26.4\% but also improves the model running time by 34.2\% compared to most advanced methods.

\end{abstract}

\maketitle
\pagestyle{plain}

\section{Introduction}

In the past few years, deep learning has dominated many fields such as computer vision and natural language processing, and continues to expand in other areas. As deep learning models are becoming more and more complex, there is a fast growing demand for efficiently training large models.
On one hand, large models not only can improve the performance on individual testing datasets, but also are able to absorb knowledge from more fields to deal with different tasks, thanks to their large numbers of parameters. Hence large models greatly improve the generalizability of the neural network models. 
On the other hand, large models avoid the requirement for different, task specific models, thus can reduce the cost significantly to a certain extent.

However, as the models and datasets are getting increasingly complex, the  development of model training hardware (e.g., GPUs) is getting more and more difficult to keep up with the pace. Consequently, computing power has become a major bottleneck for further development of large deep learning models. 
Taking computer vision tasks as an example, the number of model parameters of GoogLeNet~\cite{szegedy2015going} in 2014 was 4 million, while the SOTA Swin Transformer~\cite{liu2021swin} in 2021 has 3 billion parameters, 750 times of GoogleNet's parameters. 
NLP models can be even larger, e.g., GPT-3 has 175 billion parameters. 
As the number of model parameters is increasing at a much faster pace than the growth rate of GPU memory (e.g., from P4's 8 GB to A100's 40 GB), deep learning models have to employ smaller batch sizes in training, resulting in significantly higher communication costs and performance deterioration.

Researchers have proposed many data parallelism methods to accelerate model training by splitting datasets across multiple workers, with each worker running a complete model, which poses challenges on the workers' capacity. 
For example, BERT~\cite{devlin2018bert} requires at least more than 16~GB memory, only a few quite expensive GPUs can meet this demand at the time when the model was proposed.
However, a more preferred solution is to leverage the model parallelism, namely, split a large model into multiple smaller ones and assign them to computing devices, such that large models can still be trained even using GPUs with limited memory and/or computational resources. 
For example, an A100 GPU has 40~GB memory and 19.5~TFLOPs of computing power, while five 2080Ti GPUs have the same amount of memory and three times of computing power at only a quarter of the price. Therefore, with model parallelism, much larger models can be possibly trained on 2080Ti GPUs at much lower costs.

A key challenge in the model parallelism is how to split a large model into smaller ones and assign them to computing devices to achieve higher computational efficiency. 
Note that a neural network model can be visualized as a computational graph with nodes representing computing operations and edges representing tensor data transport. With the notion of computational graphs, the model parallelism can be translated into the problem of determining appropriate placement of computing operations onto devices, which becomes critical for improving the efficiency of training large models. Note that this problem has a huge solution space and is NP-hard~\cite{hoogeveen1994three}. For example, there are $4^{5000}$ potential allocations for a model with 5000 nodes assigned to 4 devices. 

To address this challenge, recent studies use reinforcement learning (RL) to assign operations to multiple devices to minimize the running time. 
However, due to the huge solution space, RL-based methods often take several hours to find a placement, which greatly impacts the usability. 
More recently, some of the latest works use heuristic-based algorithms to quickly find placements, which either produce poor placements or require long running time. Furthermore, the Standard Evaluation, a fundamental problem of heuristic algorithm, would behave differently when applying model parallelism for large models. However, this issue has been overlooked in recent heuristic works.

In this paper, we propose Celeritas, a fast and efficient framework, to leverage the model parallelism for large models in order to achieve high computational efficiency. 
Celeritas can efficiently divide the computational graph of the deep learning framework and quickly find an optimal device placement with significantly less resource consumption than SOTA. 
Compared with existing works that need hours of training time, Celeritas can find the optimal device placement within tens of seconds. 
Specifically, our contributions are summarized as follows.
Firstly, Celeritas employs a fast and efficient Standard Evaluation, including an innovative ordering method \emph{CPD-TOPO}, to generate model information for large models. Notably the Standard Evaluation plays an important role in perfecting the model parallelism performance yet has been overlooked by previous approaches.
Secondly, we explore the differences between model parallelism and traditional scheduling tasks~(Section~\ref{subsec:merge}), and design a proper operation fusion method for Celeritas, which significantly reduces the communication-to-computing ratio (CCR) up to 10 times~(\autoref{tab:ccr}).   
Last but not least, we deployed Celeritas to optimize multiple large models. Empirical evaluation results show that Celeritas generates parallel policies 26.4\% faster than the baselines while successfully reducing the model running time by 34.2\%. Celeritas is also effective in avoiding out of memory to generate stable results.


\vspace{-0.2in}
\section{RELATED WORK} 

\textbf{Model Parallelism} assigns parts of a model to multiple devices. 
RL-based approaches have been widely explored, which mostly take hours to generate placement strategies (such strategies may well be unstable). 
Mirhoseini et al.~\cite{mirhoseini2017device} trained an RNN-based sequence-to-sequence RL model,  taking each step's running time as a reward to compute the optimal placement, which required 17-27 hours with 80 to 160 4-GPU machines. They also proposed a hierarchical model~\cite{mirhoseini2018hierarchical} for better generalization, which yet still needed 12.5 GPU-hours to obtain the strategy. 
Spotlight~\cite{gao2018spotlight} used proximal policy optimization (PPO) but still needed 9 hours to find placements. 
Post~\cite{gao2018post} combined PPO and cross-entropy reduction to identify better placement, which still took 13.5 hours to train.

An apparent disadvantage of these methods is that they require retraining for every new model and thus result in a massive amount of computation. 
Numerous methods are proposed to address this problem.
For instance, Placeto~\cite{addanki2019placeto} applied a graph embedding network to capture potential information of the computational graph to generalize to unseen models. Nevertheless, Placeto still needed hours to seek the optimal placement even on customized simulation platform due to its MDP logic. 
SGDP~\cite{zhou2020single} proposed a more sophisticated structure by using GraphSage to learn a embedding for each operation and placing operations through a neural network based on Transformer-XL. The authors later extended SGDP to solve multiple tasks without too many additional parameters~\cite{zhou2020transferable}. They proposed a recurrent attention strategy to share model parameters in operation fusion, device placement and operation scheduling. Such a complex neural network also took up to 4 hours to achieve the best performance. 
Inspired by SGDP, Mars~\cite{lan2021accelerated} proposed a simplified segment-level sequence-to-sequence model, which needed 5 hours to find the optimal placement.

Furthermore, Baechi~\cite{jeon2020baechi} proposed a method to find a placement strategy in less than three minutes using a heuristic algorithm. However, Baechi's results are mostly equal to or worse than the strategies specified by human experts. 
In order to get better results, Pesto~\cite{hafeez2021towards} formulated the device placement problem as an integer linear program. While getting a finer placement strategy, Pesto also came at a cost of up to 50 minutes of placement time.

\textbf{Task scheduling.} Traditional task scheduling domains treat tasks as DAGs that are assigned to different processors. Yet they generally assume congestion-free communication and do not impose limits on processor capacity. Even so, the DAG scheduling problem is NP-hard. 
The mainstream approaches include list scheduling and cluster scheduling.
List scheduling is a two-step process which first determines the priority of tasks and then assigns tasks to each device in order. 
ETF (Earliest Time First)~\cite{hwang1989scheduling} selects the nodes with the earliest start time (EST) one at a time and then sequentially assigns tasks to devices with EST. 
HEFT(Heterogeneous Earliest-Finish-Time)~\cite{topcuoglu2002performance}, instead, takes $blevel$ (see Section~\ref{subsec:merge}) as the priority and employs an insertion-based strategy to seek the earliest start time of a node. 
Cluster scheduling is implemented to allocate tasks to an unlimited number of processors and includes steps of merging task nodes into clusters, assigning clusters to devices. 
LC (Linear Clustering)~\cite{kim1988general} merges all nodes on the critical path into one cluster each time, deletes all edges of the path, and repeats this step to accomplish the merging. 
DSC (Dominant Sequence Clustering)~\cite{yang1994dsc} uses a more complex strategy. For each node, it is determined whether merging with its parent node will reduce the EST of that node to produce a cluster. 
In the allocation phase, GLB (Guided Load balancing)~\cite{radulescu1998glb} allocates the cluster to the devices with the lowest current load in the order of EST. It is a good attempt to optimize the computational graph with these traditional task scheduling method. Nevertheless, the special characteristics of computational graphs, such as high CCR and huge number of nodes, memory limitations on devices, etc., pose significant challenges for these approaches.

\section{Motivation}

Existing studies have been facing with numerous challenges, which motivate us to design a fast framework for large models by leveraging the model parallelism.
Table~\ref{tab:limitation} summarizes the limitations of the existing methods.

\begin{table}[!t]
\vspace{-0.15in}
\caption{Method comparison}
\vspace{-0.1in}
\label{tab:limitation}
\small{
\begin{tabular}{|c|c|c|c|c|l}
\cline{1-5}
Method         & Baseline & Placement     & Handle  & Stability of  &  \\ 
    & Measurement &  time    &  OOM &   results & \\
\cline{1-5}
HRL & Yes                 & hours             & No         & No                   &  \\ \cline{1-5}
Baechi         & No                  & seconds           & No         & Yes                  &  \\ \cline{1-5}
Pesto          & No                  & dozens of minutes & Yes        & Yes                  &  \\ \cline{1-5}
\textbf{Celeritas}        & \textbf{Yes}                 & \textbf{seconds}           & \textbf{Yes}        & \textbf{Yes}                  &  \\ \cline{1-5}
\end{tabular}
} 
\vspace{-0.1in}
\end{table}

\noindent\textbf{Out Of Memory~(OOM)} is a common annoyance in learning-based systems. For example, we optimized a model with more than one GPU capability using the learning-based technique HRL~\cite{mirhoseini2018hierarchical}. The results illustrated in Figure~\ref{fig:grapper} show that HRL suffers from OOM. The reason is that it applies an initial strategy assuming all nodes assigned to the same device. Though HRL kept trying to avoid OOM, it only succeeded a few times. The reason is that HRL albeit sets a large penalty for OOM occurrences, it provides no solution.  To the best of our knowledge, this drawback exists in other learning-based approaches as well.

\begin{figure}[ht]
    \vspace{-0.1in}
    \centering
    \includegraphics[width=.95\linewidth]{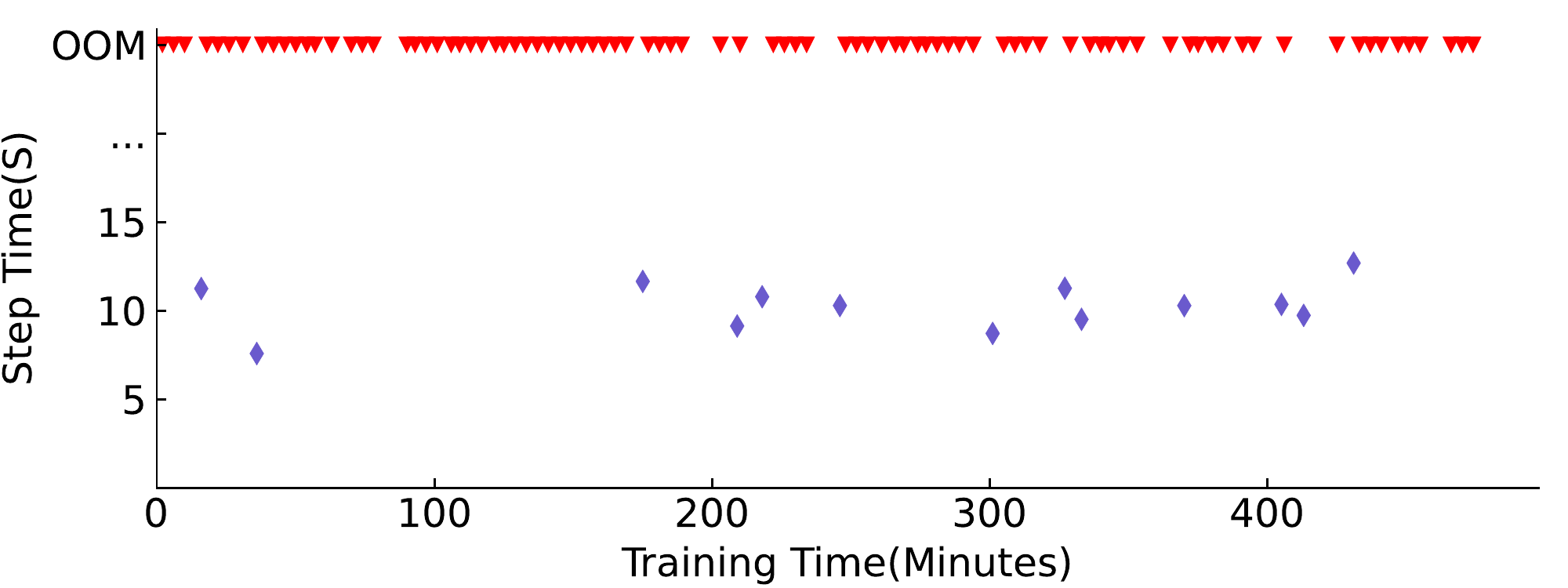}
    \vspace{-0.1in}
    \caption{The results of HRL running on the Inception\_V3 model for 8 hours with the batch size of 512. 
    } 
    \label{fig:grapper}
    \vspace{-0.15in}
\end{figure}

\begin{figure*}[!t]
\vspace{-0.15in}
    \centering
    \includegraphics[width=.75\linewidth]{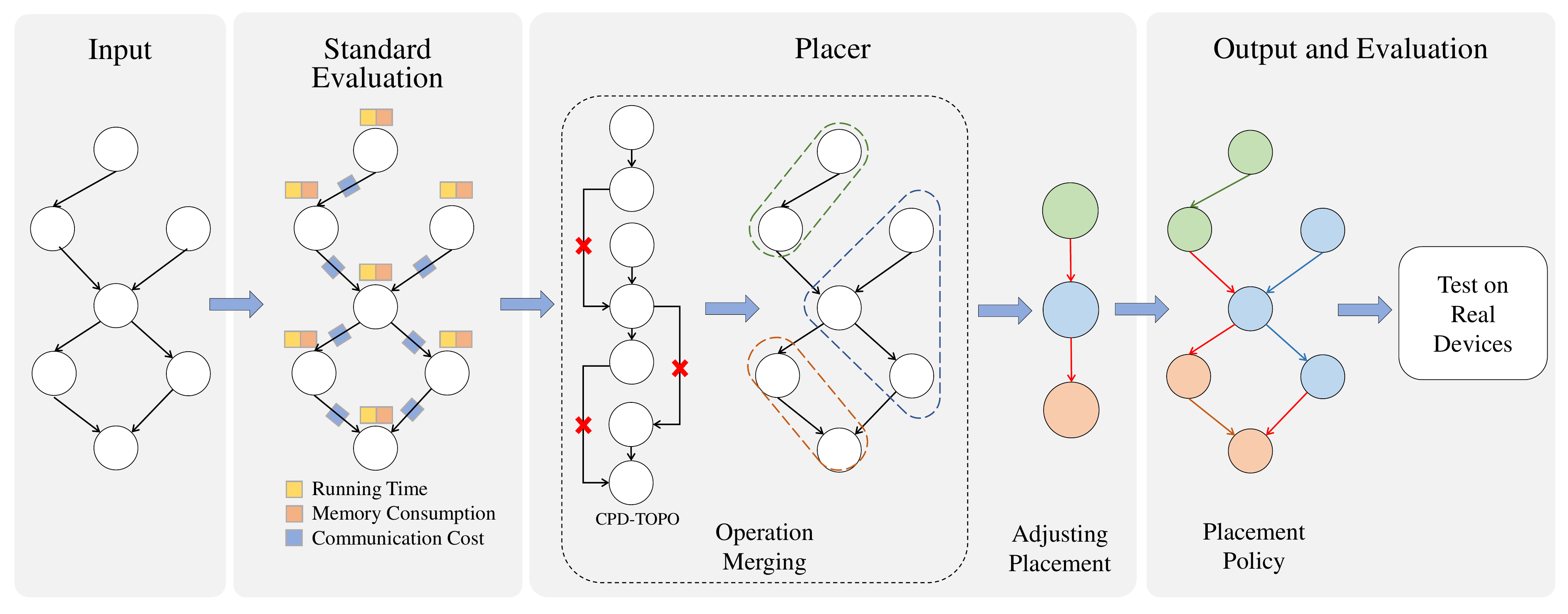}
    \vspace{-0.1in}
    \caption{Overview of Celeritas. The Standard Evaluation acquires the edge and node information of the computational graph. Celeritas then coarsens the graph by merging the CPD-TOPO ordered nodes. Afterwards, Adjusting Placement allocates the nodes in the coarse graph. The final placement strategy is generated by mapping the coarse graph to the original one. The different colors of the nodes indicate the placement on different GPUs.} 
    \label{fig:frame}
    \vspace{-0.1in}
\end{figure*}

\vskip 0.08in\noindent\textbf{Standard Evaluation.} Unlike RL-based methods, heuristic methods cannot constantly try different placement strategies on real devices. Instead, heuristic methods run the model for a handful of iterations at the beginning in order to estimate the computation time and memory consumption of the model, which is most likely to change as the batch size of the model grows. 
However, current heuristic methods do not utilize a model parallelism strategy in the Standard Evaluation. Hence their operation information is based on training on a single GPU at the beginning. Consequently, when the batch size is larger than a single GPU's memory, these heuristic methods are no longer applicable. This limitation contradicts the goal of model parallelism, training a model too large for a single GPU, thus severely impacting the practicability and flexibility.

\vskip 0.08in \noindent \textbf{Trade-off.} OOM happens occasionally with heuristics due to too large models or inaccurately estimated memory consumption (see Section~\ref{sec:experiments}). A more serious drawback of heuristic methods is the low efficiency. For instance, Baechi~\cite{jeon2020baechi} reports that their results are mostly identical to or worse than the strategies recommended by human experts. Pesto~\cite{hafeez2021towards} presents 10-30\% performance improvement with a cost of up to 50 minutes of placement time.

\section{Placement Preparation}\label{sec:prepare}

We now present Celeritas, an efficient and lightweight framework, to address the aforementioned challenges. The overall architecture of Celeritas is illustrated in Figure~\ref{fig:frame}. We break down the description into \emph{placement preparation} (Section~\ref{sec:prepare}) and \emph{placer} (Section~\ref{sec:placer}).

\subsection{Problem Formulation}

A deep learning model can be represented as a directed acyclic graph~(DAG), $G(V, E)$, where a node $v \in V$ refers to a computation operation and a directed edge $e \in E$ represents a data transmission between two nodes. 
Given a set of devices (GPUs) $D=\{d_1,d_2,...,d_m\}$, our goal is to assign a total of $n$ nodes (computational operations) in graph $G$ to the devices, in order to fulfill the training demand with minimal running time. 
This problem is NP-hard~\cite{hoogeveen1994three}. Note that the operations allocated for each device are constrained by the GPU memory, resulting in greater challenges.

\subsection{Standard Evaluation}
\label{subsec:initial}

\subsubsection{Rough Estimation.} 

Construction of the computational graph requires estimated computation time and memory consumption of the nodes, as well as the communication costs among nodes. 
As mentioned in~\cite{hafeez2021towards}, the communication costs among nodes can be determined by a linear fit: $t = kd+b$, where $d$ is the transmitted data size, and $k$ and $b$ are constants calculated in advance using off-line analysis. So we just need to focus on the node information.
This is straightforward when the batch size is small enough for the entire model to be placed on a single GPU. In such a case, we only need to run the model on a single GPU for a few iterations to collect the graphical information. 
However, most model parallelism jobs demand more memory than a single GPU supports, which has been overlooked by previous works.

Celeritas employs a new Standard Evaluation method, which consists of two steps. In the first step, we use multiple small-batch models to roughly estimate the large-batch of computational graphs; while in the second step, we acquire accurate operation information by running the large-batch model for a few iterations using a sequential placement strategy based on memory constraints. 

More specifically, in the first step, we adjust the batch size to ensure the model complies with the memory limitation of a single GPU; we then run the model for a few iterations for data collection. 
The rationale of this step comes from two key observations:
on one hand, the computation time of the model barely changes in different iterations, as observed in~\cite{mirhoseini2017device,mirhoseini2018hierarchical}; 
on the other hand, changing the batch size resulted in varying memory consumption and computation time, as shown by~\cite{gao2020estimating}, where the memory consumption is directly linked to the batch size.
To predict the nodes' memory consumption with large batch sizes, we can establish a linear regression model for each node by measuring the nodes' memory consumption with small batch sizes. Unfortunately, this method does not apply well to the computation time prediction but can only provide a rough estimation.

In the second step, we use the roughly estimated computational graph via the linear regression model (which is obtained in the first step) to seek a fast and effective method to realize model parallelism. 
The main goal is to place the model using a large batch size on as few GPUs as possible. In other words, we aim to improve the utilization of each GPU as much as possible.
We can leverage a straightforward approach, in which we first place nodes on one device until reaching its memory limit, then move on to the next device. However, this straightforward approach has many drawbacks. Instead, we adopt the topological ordering to overcome the drawbacks.

\subsubsection{DFS topological ordering.} 

Note that assigning operation nodes to devices in a random order can cause huge communication costs and straggle the model's running duration. Celeritas assigns the nodes with topological ordering prior to allocation. 

Topological ordering refers to that for the edge $(u, v)$ in graph $G$, $u$ is always in front of $v$. A common method is to repeatedly acquire the node with an indegree of 0 and then remove it along with all connected edges. 
By repeating this process, a topological order of the computational graph can be generated, which, however, is not unique. 
In addition, we find that different topological orders have great influences on the placement results. 

In a manner similar to Breadth-First-Search (BFS), M-TOPO~\cite{jeon2020baechi} maintains a queue of 0-indegree nodes and then iteratively removes the node at the head as well as all its associated edges, finally appending a new 0-indegree node to the end of the queue.

\begin{figure}[!t]
\vspace{-0.15in}
    \centering
    \includegraphics[width=.65\linewidth]{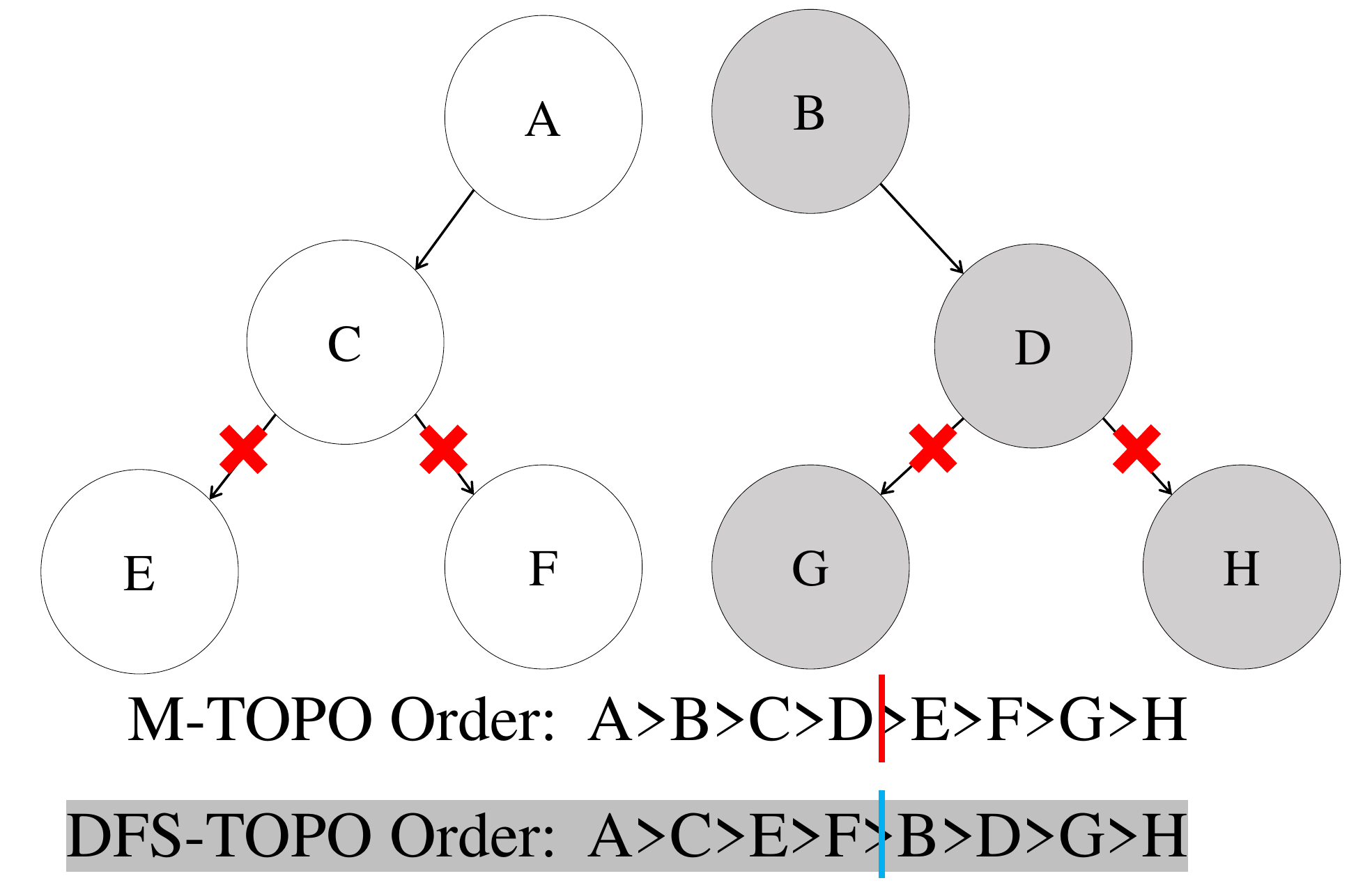}
    \vspace{-0.1in}
    \caption{M-TOPO may cause massive communication costs, while \emph{DFS-TOPO} does not.} 
    \label{fig:order}
    \vspace{-0.3in}
\end{figure}

However, a notable drawback of M-TOPO is that it ignores the connections of nodes during the sequence generation. Consequently, nearby nodes may be allocated to different devices (leading to high communication costs).
To address the above problems, we design a new approach, \emph{DFS-TOPO}, to generate topological order imitating the logic of Depth-First-Search~(DFS). 

More specifically, \emph{DFS-TOPO} maintains a queue $Q$ that contains all 0-indegree nodes. Each time, \emph{DFS-TOPO} removes $v_0$ from the head of the queue as the output of the topological sequence, and removes all the associated edges of $v_0$. 
Then the 0-indegree node among the children nodes of $v_0$ is added to the head of the queue  (see Section~\ref{subsec:merge} for multiple nodes with the same 0-indegree). As such, it provides better ordering than M-TOPO. 
For instance, consider the example of Figure~\ref{fig:order} where we would like to assign the nodes to two devices. 
Placing them in the order of M-TOPO results in cutting 4 edges, while \emph{DFS-TOPO} requires none, which suggests that M-TOPO causes higher communication costs. 
Such a case is quite common in large models. For example, around 1300 out of the 3810 nodes in the Tensor Holography model have an initial indegree of 0, which would cause enormous communication costs if using M-TOPO due to its ignoring of the connecting edges. 
\emph{DFS-TOPO} solves this problem by giving higher priority to the children nodes, resulting in much less communication costs.

%
%

\section{Placer}
\label{sec:placer}

Celeritas takes two steps to generate placement strategies: it first employs the operation fusion via CPD-TOPO to construct a coarse-grain graph, then leverages the adaptive placement via Adjusting Placement to allocate the nodes of the coarse-grain graph. 

\subsection{Operation Fusion}
\label{subsec:merge}

Traditional task scheduling methods cannot cope well with the computational graphs of large models, due to their computational graphs' large number of nodes and high communication-to-computing ratio (CCR). 
To address this challenge, we design a method of operation fusion to merge operation nodes and thus reduce the graph size. By doing so, we can effectively reduce the number of nodes and lower the CCR, thus improving the effectiveness of the placement strategy and the running time of the placement algorithm. 
More specifically, operation fusion consists of three components: CCR reduction, loop avoidance via critical path DFS-TOPO based sorting, and merging via the CPD-TOPO algorithm. Note that to avoid merge-induced loops, we sort the nodes in critical path DFS-TOPO order (a modified version of DFS-TOPO). Additionally, we use Kernighan's Algorithm to find the optimal breakpoint and merge all nodes between two adjacent breakpoints.

\subsubsection{CCR Reduction.}
\label{subsubsec:ccr}

According to ~\cite{hafeez2021towards}, most modern deep neural network models consist of thousands of nodes and the processing time of the nodes is negligible in comparison to the data transfer time between the nodes. To describe this in more detail, we leverage the communication-to-computing ratio (CCR)~\cite{wang2018list}, a notion from the task scheduling field. For the computational graph $G(V, E)$, its CCR is defined as
\begin{equation}\label{eq:CCR}
   {CCR = \frac{\sum_{(v_i,v_j)\in E}c_{i,j}}{\sum_{v_i\in V}w_i}}
\end{equation}
Where $c_{i,j}$ refers to the data transmission time from node i to node j, $w_i$ is the computation time of the computing node $i$. Table \ref{tab:ccr} shows the number of nodes and CCR of some computational graphs. Even for the simpler models, there are several thousand nodes, and the complex ones such as Transformer have up to 36352 nodes. For CCR, the lowest is 18.63, and Transformer is as high as 111.957. In traditional task scheduling, the maximum number of nodes is usually only a few hundred with a maximum CCR around 10~\cite{wang2018list}. Thus conventional task scheduling strategies lack the capability to handle large-scale computational graphs. For example, some studies have discovered that the ETF scheduling algorithm employed in Baechi performs poorly when the CCR is high~\cite{wang2018list}~\cite{djigal2020ippts}~\cite{huang2016task}.

Therefore, we try to reduce CCR by operation fusion to improve the applicability of existing scheduling algorithms. We choose to merge nodes into clusters and consider all the nodes within a cluster as an integrated new node. Naturally, the nodes merged into the same cluster will be placed on the same device and thus the data transfer time can be ignored compared to the data transfer between devices. The computation time of the new node equals the total computation time of all member nodes in the cluster. As a result, the denominator of Equation~\ref{eq:CCR} is reduced and the numerator is increased, which effectively reduces the CCR.
Another benefit of this method is that it decreases the graph size which can significantly reduce the running time. Some large computational graphs require so long time for optimization that is intolerable in practice. For instance, the Transformer variant we used, which contains 36,352 nodes, requires more than ten minutes to complete the optimization even using the simple ETF scheduling method. Pesto~\cite{hafeez2021towards} reports a running time of over a week on a computational graph with 4000 nodes. Our merging method can greatly accelerate the pace.

\subsubsection{Loop Avoidance.}
\label{subsubsec:avoid}

Note that the merging operation is not straightforward and can cause serious consequences with the wrong configuration. Because erroneous merging can create a loop in the generated cluster. In the example of Figure~\ref{fig:merging}, merging nodes $M$ and $N$ to generate a new node $X$ generates a ring $X \rightarrow Y \rightarrow Z\rightarrow X $. 
Loops can cause difficulties in both the optimization and execution of the computational graph. This is due to the fact that the nodes in the computational graph represent computational tasks each of which must wait for the predecessor node to complete computation and send the data. The nodes in the loop would be trapped in a state of waiting for each other, which is similar to the deadlock state in multithreading. In a nutshell, we must avoid loops. 
%
The following theorem shows the existence of a sufficient condition for merging nodes in computational graphs.

\begin{figure}[!t]
\vspace{-0.1in}
    \centering
    \includegraphics[width=.65\linewidth]{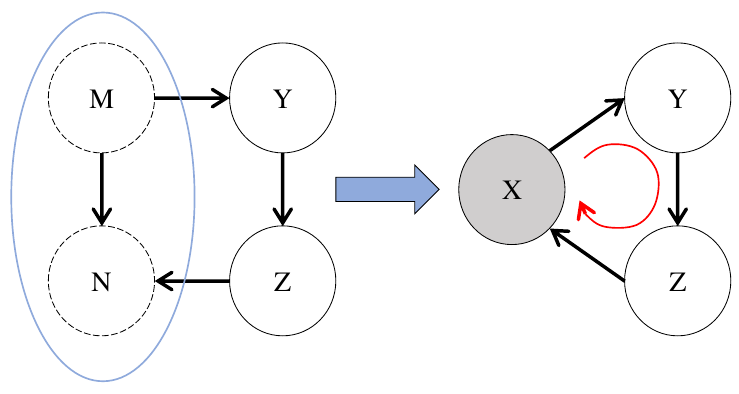}
    \vspace{-0.1in}
    \caption{Improper merging may cause loops.} 
    \label{fig:merging}
    \vspace{-0.2in}
\end{figure}

\begin{theorem}
\label{theorem:merging}
For $e=(u, v)\in E$ in a DAG, merging $u, v$ does not introduce a loop if and only if there is no other path from $u$ to $v$.
\vspace{-0.1in}
\end{theorem}
\begin{proof}
Sufficiency: assume that there exists no other path from $u$ to $v$ except the directed edge $(u, v)$. 
Suppose that merging nodes $u$, $v$ produces node $x$ and introduces a cycle $L (x \rightarrow v_1 \rightarrow v_2 \rightarrow \cdots \rightarrow v_n \rightarrow x)$. 
Since the graph before merging is acyclic, there must be at least one path $(u \rightarrow  v_1 \rightarrow v_2 \rightarrow \cdots \rightarrow v_n \rightarrow v)$, which contradicts the previous assumption. 

Necessity: when merging $u,v$ into a new node $x$ does not introduce loops, namely, there is no path $(x \rightarrow v_1 \rightarrow v_2 \rightarrow  \cdots \rightarrow v_n \rightarrow x)$; in other words, there is no other path from $u$ to $v$.
\end{proof}

However, verifying the existence of other paths from node $u$ to node $v$ is challenging due to its time complexity of $O(|E|+|V|)$. Baechi~\cite{jeon2020baechi} addresses this challenge by limiting the merged node's indegree or outdegree to 1, which avoids loops but can only merge the smaller number of nodes with the required degree. In contrast, we choose to merge the nodes in topological order and provide a guarantee to avoid loops based on the following lemma.

\begin{lemma}
\label{lemma:merging}
Merging two adjacent nodes $u$, $v$ in a topologically sorted queue does not create a loop.
\vspace{-0.1in}
\end{lemma}
\begin{proof}
Suppose there exists an edge $e = (u, v) \in E$ with topological order $O_u, O_v$ and if there is another path $(u \rightarrow v_1 \rightarrow v_2 \rightarrow  \cdots \rightarrow v_n \rightarrow v)$, then there must be $O_u > O_{v_1} > \cdots >O_v$, which is inconsistent with the adjacent order of $u$ and $v$. 
\end{proof}
 
\subsubsection{Optimal Operation Fusion.}
\label{subsubsec:cpd}

Because of the restriction that only topologically ordered adjacent nodes can be merged, the topological order directly affects the final results. Inspired by the task scheduling approach~\cite{wang2018list}, we propose an improved version of \emph{DFS-TOPO}, namely critical path DFS-TOPO (\emph{CPD-TOPO}). 

The critical path~\cite{wang2018list} is a concept in task scheduling algorithms that refers to the path from the source node to the sink node in a DAG, which has the longest overall time consumption of communication and computation. The source node indicates the node with indegree 0, and, the sink node is the node with outdegree 0. Specifically, we define top level $tlevel(v_i)$ as the longest path from any source node to node $v_i$ excluding the computation time of $v_i$, which can be recursively computed as follows.
\vspace{-0.05in}
\begin{equation}\label{eq:tlevel}
   {tlevel(v_i) = \begin{cases}\qquad {0}\qquad\qquad if\quad indegree(v_i) = 0 \  \\ \max\limits_{{ \forall v_p \in pred(v_i)}}(tlevel(v_p)+w_{p}+c_{p,i})\quad else\end{cases} }
\end{equation}
Where $pred(v_i)$ indicates the predecessors of $v_i$ and  $c_{p,i}$ refers to the communication time from node $v_p$ to node $v_i$, $w_p$ is the computation time of the node $v_p$. Note that $pred(v_i)$ includes only those nodes that transmit data to $v_i$, which means that $pred(v_i)$ has an outedge pointing to $v_i$. Similarly, we define bottom level $blevel(v_i)$ as the longest path from node $v_i$ to any sink node, including the computation time of the node $v_i$, as follows.
\begin{equation}\label{eq:blevel}
   {blevel(v_i) = \begin{cases}\qquad {w_i}\qquad\quad if\quad outdegree(v_i) = 0 \  \\ \max\limits_{{ \forall v_s \in succ(v_i)}}(blevel(v_s)+c_{i,s})+w_i\quad else  \end{cases} }
\end{equation}
Where $succ(v_i)$ indicates the successor nodes of $v_i$, that is, the nodes which are pointed to by the outedges of $v_i$. As such, $tlevel(v_i)+blevel(v_i)$ represents the longest path traversing node $v_i$. Let $cpath$ denote the overall time consumption,where $cpath(v_i) = tlevel(v_i)+blevel(v_i)$. The critical path can be chosen by selecting the node with the maximum $cpath$. 
The key idea is to select the closest nodes in topological order with the best effort to form the critical paths, to facilitate node merging afterwards to further reduce the lengths of critical paths (recall we only merge adjacent nodes in topological order). 
The motivation is that reducing the lengths of critical paths can effectively decrease the total DAG running time as shown by~\cite{wang2018list}.

For a DAG, we can preprocess the $tlevel$ and $blevel$ of all nodes before node merging. The overall flow of CPD-TOPO is: first, all nodes in the computational graph with indegree 0 are appended into a queue $Q$, and sorted according to the values of $cpath$ from largest to smallest. The first node $v_0$ at the head of $Q$ is selected as the order output of topology and removed from $Q$. The children nodes of $v_0$ are sorted according to the values of $cpath$ from smallest to largest, of which the one with indegree 0 is added to the head of $Q$ until all child nodes have been traversed. Then we continue to select the first node $v_0$ of $Q$. This process runs recursively until all nodes have been output. 
As a result, we can prioritize the node with the largest $cpath$ among the children of $v_0$, i.e., the node on the critical path.
\begin{algorithm}[t!]
        \caption{Optimal Operation Fusion}  
        \label{alg:Optimal Operation Fusion}
        \begin{algorithmic}[1] 
            \Require Computational Graph \emph{\textbf{G (V, E)}}, Devices 
            \Ensure Clusters of nodes

			\State $CPD\_Order\gets \textsc{CPD\_{Topo}}(G)$
			\State $v_1v_2...v_n\gets$ sort nodes of $G$ by increasing order of $CPD\_Order$
			\State $BreakPoint\gets \textsc{Optimal\_BP}(G,R,M, v_1...v_n)$ 
            \State $Clusters\gets$Divide $v_1...v_n$ according to breakpoints 
			\State \Return{$Clusters$}

			\Function{CPD\_{Topo}}{$G$}  
			\State for any $v \in V$, assign $tlevel$, $blevel$; calculate $Cpath$ values.
			%
			\State Create a queue $Q = \{ v | indegree(v) = 0, v \in V \}$.
			\State Sort $Q$ by decreasing order of $Cpath$ values.
			\State $CurrOrder\gets 0$
			\While{$Q$ is not empty} 
			\State $CurrNode\gets Q[0].dequeue()$ \Comment{Dequeue the head node}
			\State $CPD\_Order[CurrNode]\gets CurrOrder$ 
			\State $CurrOrder\gets CurrOrder+1$  
			\State $Children \gets successors(CurrNode)$ 
			\State Sort $Children$ by increasing order of $Cpath$ values.
			\State Delete all outEdges of $Currnode$
			\For{$Child$ in $Children$}
			\If {$indegree(Child)=0$}
			\State $Q.append\_to\_head(Child)$
			\EndIf
			\EndFor
			\EndWhile
			\State \Return{$CPD\_Order$}  
			\EndFunction  
		\Function{Optimal\_BP}{$G,R,M, v_1...v_n$}\Comment{Find the optimal breakpoints}
			\State $T(0),P(0)=0$
			\For{$v_j$ in $V$}
			\State $S(v_j)\gets $ Equation~\ref{eq:sum-cost} 
			\State $P(v_j)\gets v_i$
			\EndFor
			\State $v_k \gets v_n$, $BreakPoint\gets[]$
			\While{$v_k>0$}
			\State $v_k=P(v_k)$
			\State $BreakPoint\gets BreakPoint\cup v_k$
			\EndWhile
			\State \Return{$BreakPoint$} 
			\EndFunction
		\end{algorithmic}  
\end{algorithm}  


Sorting the nodes produced by CPD-TOPO by increasing order, we construct a sequence of nodes ($v_1>v_2>...>v_n$). Next, we utilize a graph partition algorithm namely Kernighan's Algorithm, which allows operation Fusion to reduce the maximum communication consumption. As shown in~\cite{kernighan1971optimal}, Kernighan's Algorithm can find the optimal splitting points for DAGs by traversing the nodes sequentially with a dynamic programming approach. We call some nodes as breakpoints and the nodes between adjacent breakpoints are merged into clusters. We aim to find proper breakpoints with which the merged clusters have the minimum inter-cluster communication cost. We define $cost(v_i,v_j)$ as the communication cost added by breakpoint $v_j$. In particular, assume $v_i$ is the breakpoint preceding to $v_j$, and the set of nodes between $v_i$ and $v_j$ is $set[v_i,v_j]=\{v_i,v_{i+1}.... ,v_{j-1}\}$. Then $cost(v_i,v_j)$ is the sum of the communication cost between $set[v_i,v_j]$ and the nodes behind $v_j$ (including $v_j$).
\vspace{-0.05in}
\begin{equation}\label{eq:cost}
   {cost(v_i,v_j) = \sum\limits_{i \le k \textless j, j \le m \le n}{c_{k,m}}\quad  }
\end{equation}
$cost(v_i,v_j)$ can be calculated recursively as follows:
\begin{equation}\label{eq:cost}
   {cost(v_i,v_j) = cost(v_i,v_{j-1})+\sum\limits_{j \le m \le n}{c_{{j-1},m}}-\sum\limits_{y\le  s\textless j-1}{c_{{s},{j-1}}} }
\end{equation}
Let $S(v_j)$ denote the sum of the communication cost of all the preceding nodes of breakpoint $v_j$, calculated as follows.
\begin{align}\label{eq:sum-cost}
   &S(v_j) = \min\limits_{v_i}{(S(v_i)+cost(v_i,v_j))}\quad  \\
  & \begin{array}{r@{\quad}r@{}l@{\quad}l}
s.t.&i\geq& j-R, \quad i=j-1,j-2\ldots  \notag\\
& \sum\limits_{\forall i\in [i,j)}&{Memory[v_i]}\leq M  \notag\\
\end{array}
\end{align}
Where $R$ is a parameter which determines the range of exploration, and $M$ denotes the memory limit we set for the device. Let $P(v_j)$ record the preceding breakpoint $v_i$ chosen by $v_j$.

When all the nodes have been traversed, we get all $P(v_i)$ ($1\leq i \leq n$), and $S(v_n)$. By searching $P(P(v_n))$ recursively ($P(P(v_n))$, then $P(P(P(v_n)))$, ...), we can retrieve the merged clusters by cutting the node sequence according to the breakpoints. The process of getting the breakpoints is summarized in the Optimal\_BP function of Algorithm \ref{alg:Optimal Operation Fusion}.

Algorithm \ref{alg:Optimal Operation Fusion} summarizes Optimal Operation Fusion~(Optimal Operation Fusion). Optimal Operation Fusion is able to merge as many adjacent nodes on the critical path as possible into the same cluster to reduce communication costs. Optimal Operation Fusion can also effectively reduce the number of nodes by up to 165 times and CCR by up to 10 times to significantly accelerate the placement of the large model, as shown in Table \ref{tab:ccr}.

\begin{table}[!t]
\vspace{-0.1in}
\caption{The number of nodes and \textbf{CCR} of baseline models. }
\label{tab:ccr}
\vspace{-0.1in}
\small{
\begin{tabular}{|c|cc|cc|l}
\cline{1-5}
\multirow{2}{*}{Model}                                       & \multicolumn{2}{c|}{Original}                & \multicolumn{2}{c|}{After Merging}         &  \\ \cline{2-5}
                                                             & \multicolumn{1}{c|}{Num. of Nodes} & CCR     & \multicolumn{1}{c|}{Num. of Nodes} & CCR   &  \\ \cline{1-5}
Inception\_V3                                                & \multicolumn{1}{c|}{6332}          & 33.118  & \multicolumn{1}{c|}{59}            & 8.137 &  \\ \cline{1-5}
NMT                                                          & \multicolumn{1}{c|}{25463}         & 18.630  & \multicolumn{1}{c|}{154}           & 6.372 &  \\ \cline{1-5}
Transformer                                                  & \multicolumn{1}{c|}{36352}         & 111.957 & \multicolumn{1}{c|}{220}           & 11.10 &  \\ \cline{1-5}
\begin{tabular}[c]{@{}c@{}}Tensor \\ Holography\end{tabular} & \multicolumn{1}{c|}{3810}          & 49.147  & \multicolumn{1}{c|}{37}            & 7.438 &  \\ \cline{1-5}
\end{tabular}
} 
\vspace{-0.2in}
\end{table}

Since both $tlevel$ and $blevel$ in CPD-TOPO can be pre-processed and the order of each node needs only once calculation, the time complexity of CPD-TOPO is $O(|V|+|E|)$. The complexity of computing breakpoints is linearly correlated to the number of edges in the graph, i.e., $O(|E|)$~\cite{kernighan1971optimal}. Thus the time complexity of Optimal Operation Fusion is $O(|V|+|E|)$. Nevertheless, if the exploration range $R$ is set to a large number, the process would take a long time, especially for large computational graphs in practice. 
To balance the trade-off between Celeritas's running time and the model's training time, we choose the values of exploration range parameter $R$ and memory per GPU $M$ carefully. More specifically, a too large exploration range results in a too coarse merged graph, which impacts the flexibility of placement. Meanwhile, a too small range merges too few nodes per time thus requiring a long time to generate a placement strategy. Therefore, we empirically chose $R = 200$ and $M=0.25*$ in the evaluations (see Section~\ref{sec:experiments} for details).

\subsection{Adjusting Placement}
\label{subsec:placement}

Optimal Operation Fusion can effectively merge nodes into clusters to reduce communication costs. Taking the clusters as nodes and inter-cluster links as edges, we can get a coarse-grain computational graph, which is also a DAG thanks to the loop avoidance provided by Optimal Operation Fusion (Section~\ref{subsubsec:avoid}). 
The coarse-grain computational graph has considerably fewer nodes and edges than the original graph. Placing the nodes of the coarse-grain graph by simply following sequential CPD-TOPO order (which we refer to as \emph{Order-Place}) can yield  placements superior to Baechi, as shown in Table~\ref{tab:sum}, because \emph{Order-Place} minimizes the communication cost between coarsening graph nodes.

To further improve efficiency, we propose an adaptive placement algorithm namely \emph{Adjusting Placement} as shown in Algorithm~\ref{alg:PLACE}. The main idea is to place the current node on the same device as the previous one unless otherwise reduces the total running time. 
More specifically, sequential node placement may lead to a situation where all predecessors ($pred(v_i)$) of $v_i$ already finished their computation, while some other nodes are still running on the device $d_i$. Let $pre\_t(v_i)$ denote the latest completion time of $pred(v_i)$. It takes $t_w$ time for the device to become idle, therefore the node $v_i$ has to wait for $t_w$ before it can start running on $d_i$. As the sequential placement results are good enough, we conservatively assume that if $t_w$ is larger than the communication time consumption added by placement on other devices, we can place it on another device with the earliest start time ($EST$).

Adjusting Placement first sorts the nodes of the coarse graph according to CPD-TOPO and gets the sequence $V 1>V 2>...>V n$. Afterwards, the nodes are processed in sequential order. For the current node $v_i$, we compute its $EST$ on each device. When trying to place $v_i$ on device $d_i$:
\begin{equation}\label{eq:EST}
\begin{split}
   pre\_t(v_i) &= \max\limits_{\forall v_{p} \in pred(v_i)}{Fin\_t[v_{p}]+I(d_{p})c_{p,i}}\\
   I(d_{p})&= \begin{cases}\qquad {0}\qquad if\quad d_{p} = d_i \  \\ \qquad{1}\qquad else\end{cases}\\ 
   EST(d_k)&=max(pre\_t(v_i),d\_ava[d_k])
\end{split}
\end{equation}
Where $Fin\_t[v_{p}]$ is the actual completion time of the $v_p$, $d\_ava[d_k]$ is the earliest available idle time for device $d_k$. Since the node sequence follows a topological order, all predecessor nodes are already placed when scheduling $v_i$. 
If $v_i$ and $v_{p}$ are placed on the same device, i.e., $d\_i=d\_p$, there is no communication time required in between. Otherwise, the time of data transmission $c_{p, i}$ is added. By traversing all the predecessors of $v_i$, we can get $pre\_t(v_i)$. Note that $d\_ava[d_k]$ can not simply be set as the time when the last node on $d_k$ finished its operation. Instead, Adjusting Placement checks the load on $d_k$ to see if there is a time slot to insert $v_i$ after $pre\_t(v_i)$, and set $d\_ava[d_k]$ to the start time of that slot. $EST$ gets the larger value between $d\_ava[d_k]$ and $pre\_t(v_i)$, as we must wait for all the predecessors to complete operations and data transmissions as well as $d_i$ becomes idle. If the memory of a device is fully occupied, we set its $EST$ to infinity in order to move on to the next device. The definition assumes that the previous node of $v_i$, $v_{i-1}$, is placed on device $d_k$. Let $back\_cost$ denote the maximum time required to send data back to device $d_k$ from other devices.
\begin{equation}\label{eq:EST}
   back\_cost= \max\limits_{\forall v_{s} \in succ(v_i)}{c_{i,s}} 
\end{equation}
We consider $d_i$ a better option than $d_k$ for placing $v_i$ if the following condition is satisfied.
\begin{equation}\label{eq:EST}
   EST(d_k)-EST(d_i)> back\_cost \qquad \forall d_k, d_i \in D
\end{equation}
Then node $v_i$ is placed on the device with the earliest EST, otherwise, $v_i$ is placed on $d_k$. If the memory of $d_k$ is fully occupied, the device with sufficient memory at the earliest EST is selected. If all devices are out of memory, a best-effort strategy is adopted, i.e., the device with the lowest memory usage is selected.
It is able to prove that each adjustment step of Adjusting Placement can reduce the running time or at least keep the same with \emph{Order-Place}.

\begin{algorithm}[t!]
        \caption{Adjusting Placement}  
        \label{alg:PLACE}
        \begin{algorithmic}[1] 
            \Require Coarse Graph \emph{\textbf{G' (V', E')}}, Original Computational Graph \emph{\textbf{G (V, E)}}, Devices $D$
            \Ensure Placement of G
            
			\State $CPD\_Order \gets \textsc{CPD\_{Topo}}(G')$  
		    \State $v'_1, v'_2, \cdots, v'_n \gets $ sort nodes of $V'$ by increasing order of $CPD\_Order$ 
		    \State $Ava\_m\gets$ Available memory of each device
			\For{$v_i$ in $v'_1, v'_2, \cdots, v'_n$}
			\State $BackCost \gets \max\limits_{\forall v_{s} \in succ(v_i)}{c_{i,s}}$

            \For {$d_i$ in devices $D$}:
            \If {$Ava\_m[d_i]<Memory[v_i]$}
            \State EST($d_i$)$\gets +\infty$
			\Else 
			\State EST($d_i$)$\gets$ \textsc{comp\_{EST}}($G, d_i, v_i$) 
		    \EndIf
			\EndFor
			\State $d_k\gets Placement[v_{i-1}]$
		    \State $d_1, d_2, \cdots, d_m\gets$ sort devices by increasing order of $EST$ 
			\If { $EST(d_k)-EST(d_1) >back\_cost$}  
			\State $Placement[v_i]=d_1$
			\ElsIf {$EST(d_k)< +\infty$}
			\State $Placement[v_i]=d_k$
			\Else \Comment{All devices are out of memory}
			\State $Placement[v_i] \gets \mathop{\arg\min}\limits_{d_i}Ava\_m[d_i]$
			\EndIf
			\State $Ava\_m[Placement[v_i]]\gets Ava\_m[Placement[v_i]]-Memory[v_i]$
			\State Update $Fin\_t[v_{i}]$
			\EndFor
		    \State $G\_Placement\gets$ Assign V to the devices with $Placement$ and the mapping between $V'$ and $V$
			\State \Return{$G\_Placement$}  

			\Function{comp\_EST}{$G, d_i, v_i$}  
			\State $pre\_t(v_i)\gets 0$
			\For{$v_{p}$ $\in$ $predecessors(X_i)$}  
			\If {$d_{v_{p}} = d_i$}
			\State $pre\_t(v_i)=\max{(pre\_t(v_i),Fin\_t[v_{p}])}$
			\Else
			\State $pre\_t(v_i)=\max{(pre\_t(v_i),Fin\_t[v_{p}]+c_{p,i})}$
			\EndIf
			\EndFor
			\State $d\_ava[d_i]\gets$ earliest time $v_i$ can be placed on $d_i$
			\State $EST(d_i)=max(pre\_t(v_i),d\_ava[d_i])$  
			\State \Return{$EST(d_i)$}  
			\EndFunction  
		\end{algorithmic}  
\end{algorithm}

\textbf{Proof}: Assuming that we use Order-Place to place the current node $v_i$ and the start computation time of $v_i$ is $t\_order$. Then the earliest start time of all the successors ($succ(v_i)$) is $t\_order+w_i$, where $w_i$ is the time required for $v_i$ to run. If Adjusting Placement chose to place $v_i$ on another device $d_m$, then we must have $EST(d_k)-EST(d_m) >
 back\_cost$. At this point the earliest start time of the successors of $v_i$ can be represented as $t\_{adjust}$. We have:
 \begin{equation*}
\begin{split}
t\_{adjust}&<EST(d_m)+back\_cost+w_{i}\\
&<EST(d_k)+w_{i}\\
&<=t\_order+w_{i}
\end{split}
\end{equation*}
 
Since Adjusting Placement only traverses each node once and checks every incoming and outgoing edge, its time complexity is $O(|V|+|E|)$. Because the coarse graph has far fewer nodes than the original graph, i.e., a few hundred compared to thousands as shown in Table~\ref{tab:ccr}, Adjusting Placement runs very fast.

\definecolor{Gray}{gray}{0.9}

\begin{table*}[!t]
\caption{The average single-step training time (lower is better) in seconds on 4 GPUs.}
\label{tab:sum}
\vspace{-0.15in}
\small{
\begin{tabular}{|c|c|ccc|c|c|c|c|l}
\cline{1-9}
\multirow{2}{*}{Model} & \multicolumn{1}{l|}{\multirow{2}{*}{Metis}} & \multicolumn{3}{c|}{Baechi}                                                              & \multirow{2}{*}{HRL} & \multirow{2}{*}{Order-Place} & \multirow{2}{*}{\textbf{Celeritas}} & \multirow{2}{*}{Speed up} &  \\ \cline{3-5}
                       & \multicolumn{1}{l|}{}                       & \multicolumn{1}{c|}{m-TOPO} & \multicolumn{1}{c|}{m-ETF} & \multicolumn{1}{l|}{m-SCT} &                                 &                               &                          &                           &  \\ \cline{1-9}
Inception\_V3          & OOM                                         & \multicolumn{1}{c|}{3.100}   & \multicolumn{1}{c|}{10.548} & 7.610                       & 8.565                           & 2.931                         &\textbf{2.859}                    & >7.8\%                     &  \\ \cline{1-9}
NMT                    & OOM                                         & \multicolumn{1}{c|}{OOM}     & \multicolumn{1}{c|}{OOM}    & OOM                         & 10.224                          & 9.117                         & \textbf{9.099}                    & >11.0\%                    &  \\ \cline{1-9}
Transformer            & OOM                                         & \multicolumn{1}{c|}{OOM}     & \multicolumn{1}{c|}{8.613}  & 8.411                       & 42.837                          & 6.935                         & \textbf{6.536}                    & >22.3\%                    &  \\ \cline{1-9}
Tensor Holography            & 6.591                                       & \multicolumn{1}{c|}{2.220}   & \multicolumn{1}{c|}{6.376}  & 6.013                       & 3.683                           & 1.520                         & \textbf{1.461}                    & >34.2\%                    &  \\ \cline{1-9}
\end{tabular}
} 
\end{table*}

\section{EXPERIMENTS}
\label{sec:experiments}

\subsection{Implementation}
\label{sec:implementation}

We implement Celeritas with the TensorFlow framework and use TensorFlow's profiler module to monitor the computational graph (for instance, obtaining the operation start time and the output tensors). 
Moreover, TensorFlow has co-location restrictions, i.e., nodes within a co-location group need to be placed on the same device. We adopt a scheme similar to Baechi, namely, if the first node of a co-location group is placed on device $d$, then all nodes of this group are placed on $d$.

We also consider potential congestion caused by simultaneous communications with multiple devices for each device. Inspired by the mechanism deployed by TensorFlow for communications between multiple devices~\cite{abadi2016TensorFlow}, we treat the data transmission process as an additional task, by adding a new operation node for data transmission on the devices which need to send or receive data. 

\subsection{ Experimental Setup}

\subsubsection{Testbed setup.} 
We deployed Celeritas on a server with 2 Intel Xeon Gold 5220 CPUs and 4 NVIDIA Tesla V100 GPUs (each has 32GB memory). GPUs are interconnected via the PCIe bus. Our experiments are based on TensorFlow 1.12.

As mentioned earlier, model parallelism is mostly used in cases where a model exceeds the memory of a single GPU. We find that previous works focus most of their efforts on optimizing small models. To evaluate the advantages of Celeritas, we adopt large models with large batch size settings in the experiments.

\vskip 0.02in \noindent \textbf{Inception\_V3}~\cite{szegedy2016rethinking} is a CNN model well known for tasks such as image classification. 
We set its batch size to $512$, which exceeds a single GPU's capacity.

\vskip 0.02in \noindent  \textbf{NMT}
~\cite{wu2016google} is a sequence-to-sequence model with multiple LSTM units and is mostly used for machine translation. We use a 4-layer model variant of NMT with 2048 hidden units per LSTM and the batch size is also set to $512$.

\vskip 0.02in \noindent  \textbf{Transformer}~\cite{vaswani2017attention} is another well known attention-based NLP model consisting of an encoder and a decoder. We use a Transformer variant with $12$ layers and $16$ heads containing $2048$ hidden units. We set its batch size to $256$.

\vskip 0.02in \noindent  \textbf{Tensor Holography}~\cite{shi2021towards} is designed to synthesize realistic color 3D holograms from a single RGB depth image. It is composed of $30$ convolutional layers and $24$ filters per layer. The authors trained it for $84$ hours on an NVIDIA Tesla V100 GPU to get the final result. It has a huge footprint of memory usage, which requires 30 GB of GPU memory even with a batch size of $16$. We set its batch size to $32$ in our experiments.

\vspace{-0.05in}
\subsubsection{Baseline.}
To evaluate Celeritas, we compared our results with the following baselines.  

\vskip 0.02in \noindent \textbf{Metis}~\cite{karypis1997metis} is widely used in the field of graph partitioning. It attempts to minimize edge partitioning using k-way partitioning while satisfying a balance limit associated with the weights. 

\vskip 0.02in \noindent \textbf{HRL}~\cite{mirhoseini2018hierarchical} generates placement strategies with an RL-based approach. A feedforward neural network assigns groups of operations and then an LSTM-based model is applied to produce a placement strategy for each group. The resulting strategy is measured on the actual device to obtain the running time, which is used as the reward for optimization. According to~\cite{lan2021accelerated}, HRL can find the optimal placement within 5 hours, and we use this setup in our experiments.

\vskip 0.02in \noindent \textbf{Baechi}~\cite{jeon2020baechi} proposes a series of heuristics algorithms, including m-TOPO, m-ETF, and m-SCT, of which m-SCT performs the best. m-ETF places the operations on the device with the earliest start time. m-SCT improves m-ETF by using integer linear programming to find the operation's most favorite child. Since Baechi lacks the initial information measurement for large models, we utilized Celeritas's Standard Evaluations instead throughout the experiments.

For comparison, we also add \textbf{Order-Place}, as described in Section~\ref{subsec:placement}. More specifically, for the coarse graphs created by operation fusion, we place nodes in CPD-TOPO order on one device until reaching its memory limit, then move on to the next device.

\vspace{-0.1in}
\subsection{Evaluation Methodology}

We adopt a set of performance metrics to evaluate the algorithms.

\vskip 0.02in \noindent \textbf{Single-step Time.} 
We run the model according to the placement strategy generated by the methods and measure single-step training time. As suggested in \cite{mirhoseini2018hierarchical}, we discard the running time of the first five steps in the measurement, in order to avoid the variations during warming up. Different from \cite{mirhoseini2018hierarchical}, we measure the average training time for the next fifty steps rather than just the first five for more robust results. 

\vskip 0.02in \noindent \textbf{Placement Generation Time.} 
We measure the time each method took to generate the final placement, excluding the time of Standard Evaluation and the final measurement of the placement that is required by all approaches.

Additionally, since Standard Evaluations play a key role in determining the placement strategies, and none of the previous works have considered Standard Evaluations, we evaluate the efficiency of Standard Evaluations using the following two metrics.

\vskip 0.02in \noindent \textbf{Estimation Accuracy.} 
We evaluate the accuracy of the memory consumption estimation by comparing the gap between the computational graph information generated by Celeritas estimation and the real measurements.
More specifically, for each node in the computational graph, we calculate the relative deviation between the estimated information $d\_e$ and the actual measured information $d\_a$, i.e., $| d\_e- d\_a|/ d\_a$. 
%

\vskip 0.02in \noindent \textbf{Measurement Time.} 
We also evaluate the impact of different placement methods on the Standard Evaluation time (namely, the time to complete the Standard Evaluation). We compare m-TOPO and DFS-TOPO defined in Section~\ref{subsec:initial} (Celeritas employs CPD-TOPO, an augmented version of DFS-TOPO). To avoid OOM, we also add a best-effort strategy for both methods: if all devices are out of memory, the device with the lowest memory usage is selected. 
Note that we discard the running times of the first 5 steps and instead use the average measurement time of the next 50 steps.

\subsection{Performance}

\subsubsection{Single-step Time.} 
The results of step times for all algorithms are shown in Table \ref{tab:sum}.
The column of \emph{speed up} shows the acceleration Celeritas provides compared to the second best alternative. Celeritas outperforms all baselines on all tested models with considerable improvements.

We make the following observations.
First, Inception\_V3 has fewer nodes than most others but still requires memory more than a single GPU's capacity. 
The resulting OOM has plagued HRL in all experiments. Celeritas is 62.4\% faster than m-SCT in single-step time and 7.8\% faster than m-TOPO, the best method in Baechi. 

Second, NMT demands the largest memory among the tested models. HRL achieves a satisfactory result through many attempts. M-SCT and m-ETF encounters OOM during placement. Celeritas successfully avoids OOM thanks to its best-effort placement strategy, i.e., assigning the nodes that cannot find the optimal placement to the device with the lowest memory usage. Celeritas's single-step training time is 11.0\% faster than HRL.

Third, Transformer has a large number of nodes, which brings a significant challenge for optimization. For example, HRL required a super long single-step running time. Celeritas benefits from the node merging process and is able to get the placement 22.3\% faster than m-SCT. Meanwhile, Adjusting Placement gives Celeritas a 5.8\% improvement compared to Order-Place. 

Last, Tensor Holography contains fewer nodes but also consumes a larger amount of memory. Metis prevents OOM but does not perform as well as others. Celeritas reduces single-step time by 34.2\% compared to the second best performed algorithm, m-TOPO.

\begin{figure*}[t!]
\begin{minipage}[][][b]{0.66\textwidth}
    \centering
     \begin{subfigure}{0.45\textwidth}
        \includegraphics[width=\linewidth]{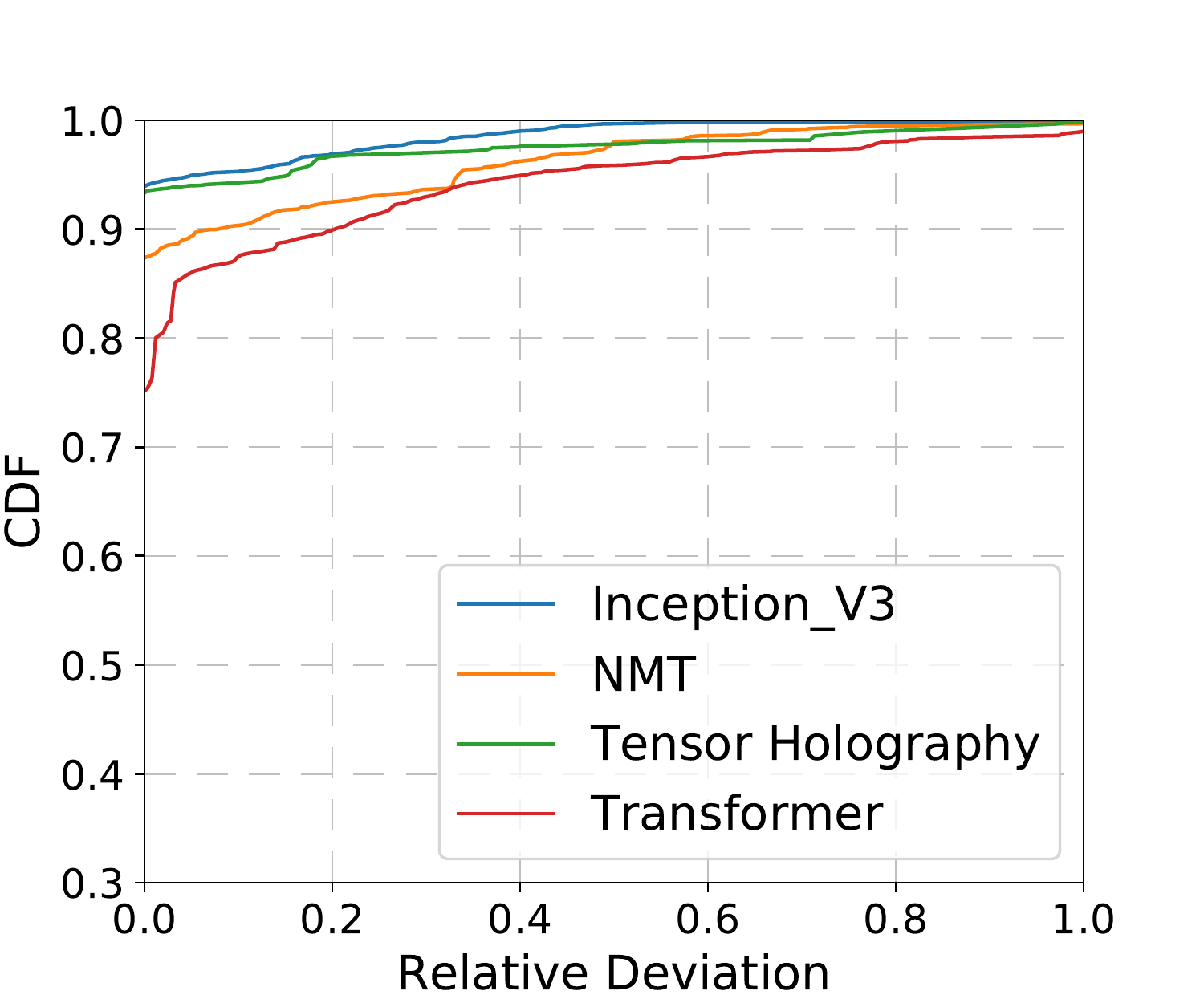}
         \caption{Memory Deviation to CDF}
         \label{fig:cdf-mem}
     \end{subfigure}
     \begin{subfigure}{0.45\textwidth}
         \includegraphics[width=\linewidth]{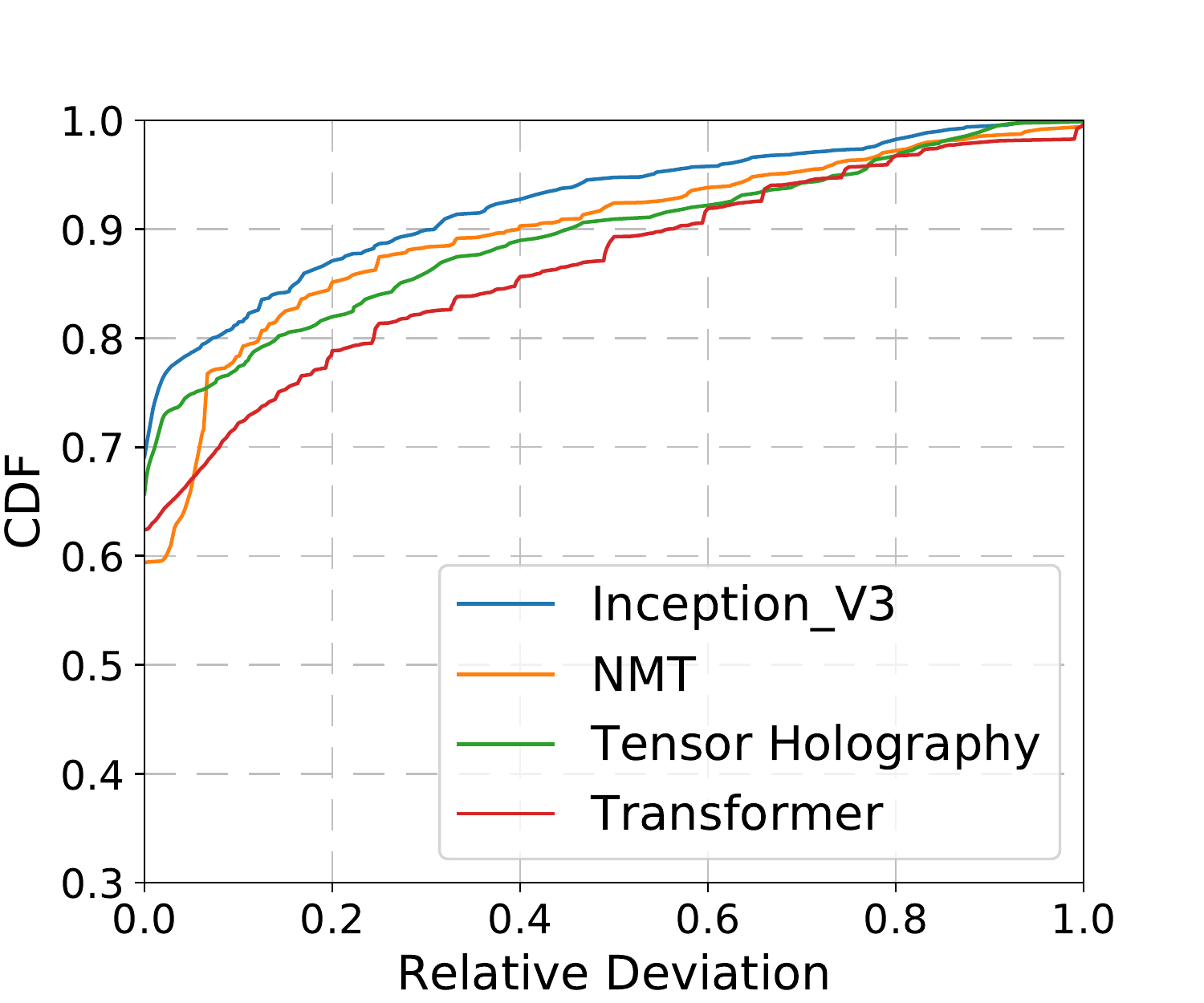}
    \caption{Time Deviation to CDF}
    \label{fig:cdf-time}
    \end{subfigure}
    \vspace{-0.1in}
    \caption{CDF of the relative deviation between estimated and accurate information.}
    \label{fig:cdf}
\end{minipage}
\begin{minipage}[][][b]{0.33\textwidth}
    \vspace{0.2in}
    \includegraphics[width=\textwidth]{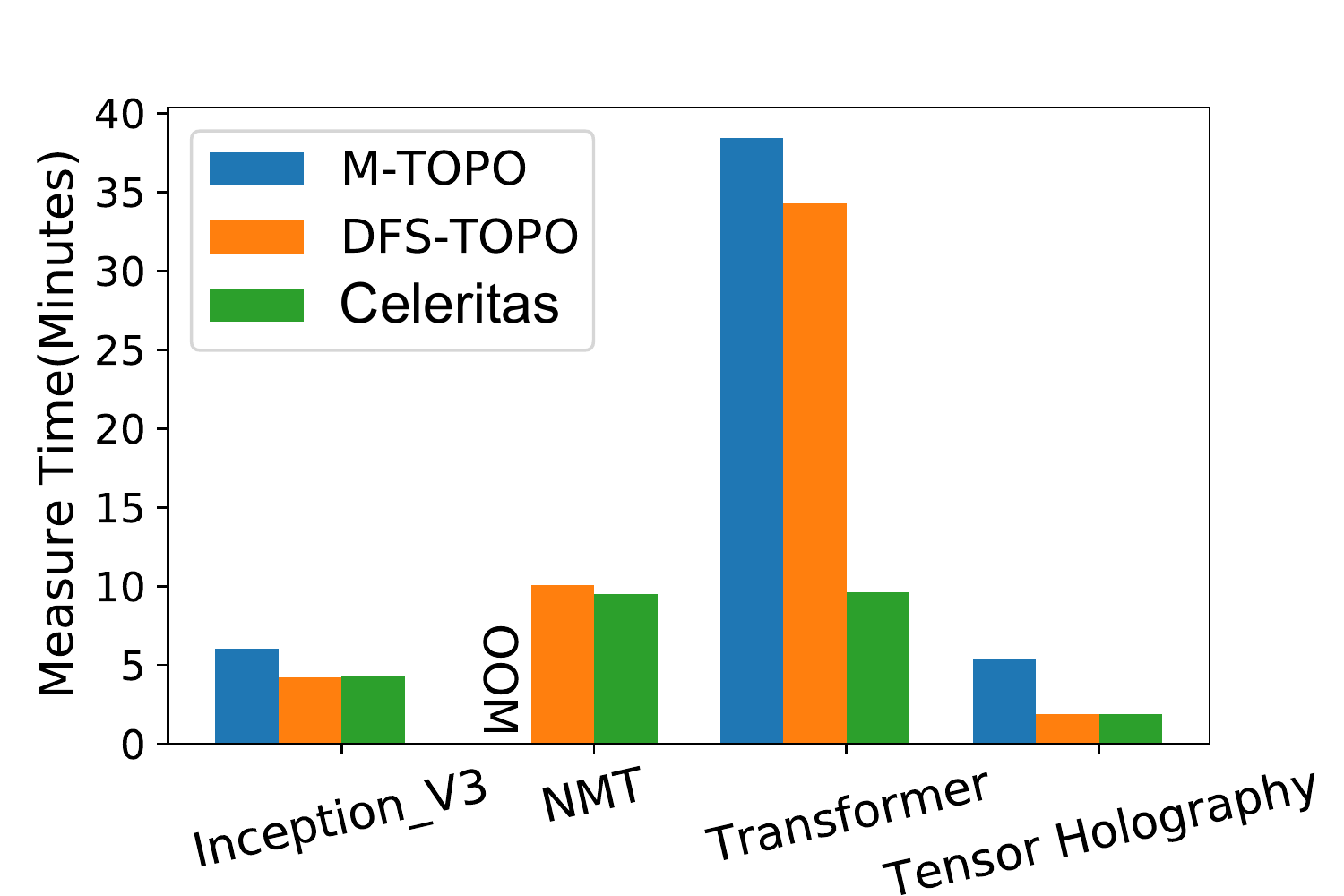}
    \vspace{0.01in}
    \caption{Standard Evaluation time.}
    \label{fig:measure}
\end{minipage}
\vspace{-0.15in}
\end{figure*}


\vspace{-0.05in}
\subsubsection{Placement Generation Time.} 
Table \ref{tab:time} summarizes the results.
Note that for HRL, we utilized 5 hours reported by Mars~\cite{lan2021accelerated} for training, 
and for Baechi, we chose m-SCT which reported the best results in~\cite{jeon2020baechi}.

\begin{table}[!t]
\vspace{-0.1in}
\caption{Time (seconds) for placement strategy generation.}
\label{tab:time}
\vspace{-0.1in}
\small{
\begin{tabular}{|c|c|c|c|c|l}
\cline{1-5}
Model         & m-SCT  & HRL & Order-Place & \textbf{Celeritas} &  \\ \cline{1-5}
Inception\_V3 & 8.591   & 5 hrs        & 5.439        & \textbf{7.183}   &  \\ \cline{1-5}
NMT           & -       & 5 hrs        & 54.695       & \textbf{59.418}  &  \\ \cline{1-5}
Transformer   & 129.498 & 5 hrs        & 92.639       & \textbf{95.304}  &  \\ \cline{1-5}
Tensor Holography     & 2.968   & 5 hrs        & 2.711        & \textbf{3.586}   &  \\
\cline{1-5}
\end{tabular}
} 
\vspace{-0.1in}
\end{table}

We observe that for the models with fewer nodes such as Inception\_V3, m-SCT runs reasonably fast but Celeritas still generates the placement 16.4\% faster, 
and for Tensor holography, Celeritas takes 0.618\,s more than m-SCT, mainly due to the time of coarse graph generation. 
When the number of nodes increases, the running time also goes up. 
For placing NMT, m-SCT encounters OOM while Celeritas only takes 59.418\,s. Placing Transformer is challenging due to its 30,000+ nodes, which takes m-SCT 129.498\,s while Celeritas is 26.4\% faster, thanks to its node merging method.
We notice that, compared to Order-Place, Celeritas takes additional 2.5\,s on average, which is mainly caused by the processing time of coarse graph generation. 

\subsection{Standard Evaluation}

\subsubsection{Estimation Accuracy.}

Table~\ref{tab:bias} summarizes the average deviation results. 
On models with fewer nodes like Inception\_V3 and Tensor Holography, the average deviations of the memory estimation are within 3\%. 
The deviation for NMT rises slightly to 3.93\%. Due to the numerous nodes and complex structure of the Transformer, the average memory deviation goes up to 6.02\%. 
The estimation of the running time of the nodes is a bit harder because the correlation between running time and batch size is not quite clear. However, the results are acceptable. 
On Inception\_V3, the average time estimation deviation is only 7.99\%. And the deviation can be controlled to about 10\% on NMT and Tensor Holography. For the most complex model, Transformer, the time estimation deviation is only 14.16\%.

To further analyze the results in more detail, we plot the CDFs of the estimated relative deviations in Figure~\ref{fig:cdf}. We make the following observations.
First, as shown in Figure~\ref{fig:cdf-mem}, the CDF corresponding to 0\% relative deviation of memory estimation is high, indicating that our estimation is accurate for most nodes. 
For simple computational graphs like Inception\_V3 and Tensor Holography, the ratio of accurately estimated nodes is $>90\%$. And that ratio for NMT and Transformer is also $>70\%$. For all models, the relative deviation of memory estimation is within 20\% for 90\% of the nodes. 
Second, Figure~\ref{fig:cdf-time} shows the relative deviation of the running time estimation. The ratio of nodes with accurate time estimation is more than 60\% for all models, except for NMT, which is slightly lower. The deviation is within 30\% for 80\% for all models.

\begin{table}[!t]
\vspace{-0.1in}
\caption{Average relative deviation between estimated and accurate information for all nodes.}
\label{tab:bias}
\vspace{-0.15in}
\small{
\begin{tabular}{ccccc}
\hline
\multicolumn{1}{|c|}{}                                                                  & \multicolumn{1}{c|}{Inception\_V3} & \multicolumn{1}{c|}{NMT}     & \multicolumn{1}{c|}{Transformer} & \multicolumn{1}{c|}{\begin{tabular}[c]{@{}c@{}}Tensor \\ Holography\end{tabular}} \\ \hline
\multicolumn{1}{|c|}{Mem Dev.} & \multicolumn{1}{c|}{1.54\%}        & \multicolumn{1}{c|}{3.93\%}  & \multicolumn{1}{c|}{6.02\%}      & \multicolumn{1}{c|}{2.92\%}                                                       \\ \hline
\multicolumn{1}{|c|}{Time Dev.}   & \multicolumn{1}{c|}{7.99\%}        & \multicolumn{1}{c|}{10.85\%} & \multicolumn{1}{c|}{14.16\%}     & \multicolumn{1}{c|}{11.32\%}                                                      \\ \hline
\multicolumn{1}{l}{}                                                                    & \multicolumn{1}{l}{}               & \multicolumn{1}{l}{}         & \multicolumn{1}{l}{}             & \multicolumn{1}{l}{}                                                              \\
\multicolumn{1}{l}{}                                                                    & \multicolumn{1}{l}{}               & \multicolumn{1}{l}{}         & \multicolumn{1}{l}{}             & \multicolumn{1}{l}{}                                                             
\end{tabular}
} 
\vspace{-0.52in}
\end{table}

\vspace{-0.1in}
\subsubsection{Measurement Time.}

Figure~\ref{fig:measure} shows the measurement times for all methods.
For Inception\_V3, m-TOPO takes 6.0 minutes to complete the measurement, while DFS-TOPO takes only 4.22 minutes. 
Celeritas takes 6 seconds longer than DFS-TOPO because it has more processes such as merging, building coarse maps, and placing, 
while DFS-TOPO only conducts sequential placement. Similarly for Tensor Holography, Celeritas is 1 second slower than DFS-TOPO and 65.1\% faster than m-TOPO. On NMT, m-TOPO encountered OOM, which is mainly due to the topological order employed by m-TOPO and the co-location restrictions mentioned in Section~\ref{sec:implementation}, similarly to the problem m-SCT and m-ETF shown in Table~\ref{tab:ccr}. We find that nodes in the same co-location group are mostly neighbor nodes that would be placed on the same device. However, the BFS-like topological order adopted by m-TOPO separates the neighbor nodes far apart. 
M-TOPO calculates only the memory usage of the first node in the co-location group while ignoring others. Hence it can allocate too many groups of nodes connected to the neighbor nodes on a single GPU. DFS-TOPO takes only 10.07 minutes on NMT and Celeritas is even faster which takes 9.50 minutes. For the most complex Transformer, simple strategies like m-TOPO and DFS-TOPO both incur massive communication costs, taking 38.45 and 34.3 minutes to complete the measurement, respectively.  
Thanks to Optimal Operation Fusion and Adjusting Placement, Celeritas takes only 9.58 minutes, which is less than a third of DFS-TOPO.

\section{Conclusion}

In this paper, we present Celeritas, a fast and efficient model parallelism solution for model placement on multiple devices. In particular, we consider the problem of Standard Evaluation for large models and propose an efficient and complete method. We have implemented Celeritas for multiple representative large models and demonstrated its significant improvements in model training time by up to 34.2\% faster, and 26.4\% shorter placement time compared to most advanced existing approaches. Thus Celeritas significantly improves the availability and flexibility of model parallelism.

\bibliographystyle{ACM-Reference-Format}
\bibliography{reference}

\end{document}